\documentclass[12pt,a4paper]{iopart}
\usepackage{iopams}
\usepackage{amssymb}
\expandafter\let\csname equation*\endcsname\relax%
\expandafter\let\csname endequation*\endcsname\relax
\usepackage{amsmath}
\usepackage{amsthm}
\usepackage{amsfonts} 
\usepackage{mathtools}
\usepackage{calrsfs}
\usepackage{cite} 
\usepackage{color}
\usepackage{graphicx}
\usepackage{subfigure}
\usepackage{cases}

\theoremstyle{plain}
\newtheorem{thm}{Theorem}[section]

\newtheorem*{cor}{Corollary}

\providecommand{\keywords}[1]{\textbf{\textit{Keywords---}} #1}

\begin{document}

\title[]{Global Minimality in Constrained Inverse Source Problems for Metamaterials}

\author{Mohamed R. Khodja} \address{Department of Electrical Engineering \\ Prince Mohammad Bin Fahd University, Khobar 31952, Saudi Arabia} \ead{mkhodja@pmu.edu.sa}


%

\begin{abstract}
Global minimality, boundedness, and uniqueness are established for a general, physically motivated class of inverse source problems in non-homogeneous electromagnetic media with generalized constitutive parameters. The existence of a solution was addressed earlier. The radiating source, represented by the current density, was reconstructed earlier by minimizing its $L^{2}$-norm constrained to produce a prescribed radiated field while ensuring vanishing reactive power. Using the $L^{2}$-norm allows for an analytically tractable measure of the physical resources of the source, while the reactive power constraint maximizes transmitted power. Numerical study suggests that sources within active metamaterial substrates can have remarkable tuning behaviors. Tuning stability can be achieved along specific permittivity and permeability curves on zero-reactive power plots. Each permittivity or permeability value can correspond to a discrete set of dual parameter values that enable effective tuning. The tuning characteristics observed suggest that double-positive (DPS) and double-negative (DNG) substrates are more favorable for tuning than single-negative (SNG) materials, possibly due to interactions dominant in DPS and DNG media. These results fill an analytical gap in the solution of a problem that is both intriguing and challenging due to its general formulation, which requires minimizing an objective functional with nonconvex functional constraints on an unbounded domain. They also offer numerical insights that may have implications for the design and optimization of sources in complex media, which is a topic of significant current interest.
\end{abstract}


\pagebreak
\section{Introduction}
\label{Sec-intro}
How does extending the domain of physical parameters that influence wave propagation in a medium affect wave behavior? Which extensions remain consistent with physical laws, and how can these generalized parameters be practically realized? These questions, and related others, have been intermittently explored during the 20th century \cite{Mandelstam1945,SchelkunoffFriis1952,veselago}, but saw a resurgence of interest around the early 21st century with the discovery of ``metamaterials'' \cite{PendryHoldenRobbinsEtAl1999,SmithPadillaVierEtAl2000,DingLiuQiuEtAl2007}. Metamaterials are artificial, composite structures that exhibit effective wave-propagation properties unattainable in naturally occurring materials. The study of metamaterials has evolved into a vibrant interdisciplinary field encompassing areas such as nanophysics, electromagnetics, electromechanics, and applied mathematics. Much is to be done to explore what is theoretically possible and practically achievable. 

Several types of metamaterials have been identified including electromagnetic \cite{ZiolkowskiEngheta2006}, mechanical \cite{GarciaHennion2019}, and quantum metamaterials \cite{ZagoskinFelbacqRousseau2016}, each associated with a range of exotic phenomena. For instance, electromagnetic and mechanical metamaterials open up possibilities such as cloaking \cite{RomainFleuryMonticoneAlu2015,BueckmannThielKadicEtAl2014}, sub-wavelength imaging \cite{HaxhaAbdelMalekOuerghiEtAl2018,ZhuChristensenJungEtAl2011}, enhanced or suppressed radiation and scattering \cite{JinHe2010,PengChenWuEtAl2016}. Metamaterials are poised to revolutionize fields as diverse as telecommunication, remote sensing, medical imaging, photonics, photovoltaics, geophysical exploration, and nanofabrication.

Our focus here in on electromagnetic metamaterials, distinguished by their  generalized electric permittivity $\epsilon$ and magnetic permeability $\mu$. These constitutive parameters are not only complex but can also have real parts that are simultaneously negative, imbuing the equations with both mathematical and physical richness \cite{ZiolkowskiEngheta2006}. This results in wave behavior absent in conventional media including left-handed propagation, negative refraction \cite{AydinaGuven2005}, reversed Cherenkov effect \cite{DuanTangWangEtAl2017}, and the inverse Doppler effect \cite{SeddonBearpark2003}. All of these phenomena have been experimentally observed.

Two complementary methodologies guide the study of metamaterials: (1) a direct approach, where predefined, typically achievable metamaterial properties are assumed and their effects on wave propagation is investigated, and (2) an inverse-theoretic approach, which begins with a desired wave-metamaterial interaction. In the inverse-theoretic approach, constrained optimization is employed to derive the metamaterial's properties—geometry, dispersion relations, or constitutive parameters—that can achieve the desired interaction if they are physically plausible. The inverse-theoretic approach holds particular appeal in theoretical studies because its solutions encompass the broader set of \textit{physically conceivable} designs, rather than the narrower subset of \textit{practically achievable} ones. This approach continues to drive research in electromagnetic and elastic inverse problems across mathematics \cite{Fedele_et_al-2024,ArensJiLiu2020,JiLiu2019,LiChenLi2017,DengLiuUhlmann2017,CaroZhou2014,Mederski2015,LassasZhou2016,Uhlmann2014}, physics \cite{Chen_et_al-2020,MoleskyLinPiggottEtAl2018,MulkeyDilliesDurach2017,LiuGabrielliLipsonEtAl2013,RezaDehbashiaBialkowskiAbbosh2017}, and engineering \cite{Ha_et_al-2023,DenisovaRezvov2012,OkhmatovskiAronssonShafai2012,OtomoriAndkjaerSigmundEtAl2012} communities.

In \cite{MKB_SIAM_2008} the authors explored the analytical solution of a general, inverse electromagnetic-radiation problem: reconstructing a source (i.e., current density) embedded in non-homogeneous media with generalized constitutive parameters, radiating a prescribed exterior field. This reconstruction involved minimizing  the “source energy” of the current density, subject to a vanishing reactive power condition. “Source energy” refers to the $L^{2}$-norm of the current density, a metric used not to measure actual energy but for analytical tractability and as a proxy for minimizing the source’s physical resources, since the $L^{2}$-norm essentially reflects the physical energy. Reactive power in this context refers to the cycling power between the inductive and capacitive elements of the source, and minimizing it maximizes the power transmitted to the far-field. The problem is both intriguing and challenging due to its general formulation, which requires minimizing an objective functional with nonconvex functional constraints on an unbounded domain. The authors have already answered the initial question concerning the existence of a tuned (i.e., zero-reactive-power) minimal source. However, they have not addressed its boundedness, minimality, or uniqueness. 

This paper revisits this problem to address these gaps, establishing the lower boundedness, global minimality, and uniqueness of the tuned minimal source. It also presents a numerical study to give an intuitive understanding of the properties of this source in a few typical cases.    
\section{Preliminaries}
\label{Sec_Prelim}
For the sake of completeness, we give in this section a summary of the essential definitions, explicit forms, and proof of existence of the tuned minimal sources. For details, the reader is referred to \cite{MKB_SIAM_2008}. The new results of this paper, namely, the lower boundedness, global minimality, and uniqueness of the tuned minimal sources are presented in section~\ref{Sec_Results}.
\subsection{Radiating Source}
The electromagnetic source to be reconstructed is assumed to be embedded in a substrate with volume $V\coloneqq\{\mathbf{r}\in\mathbb{R}^3:r\coloneqq\|\mathbf{r}\|\leq a < \infty\}$. The electric permittivity and magnetic permeability distributions are of the form $f(\mathbf{r}) =f\theta(a-r)+f_{0}\theta(r-a)$, where $f=\epsilon,\mu$, stands for the permittivity and the permeability, respectively; subscript $0$ refers to the vacuum; and $\theta$ denotes the Heaviside unit step. For a lossless substrate $\left(\epsilon,\mu\right) \in \mathbb{R}^2$ and, thus, the radiation propagation constant $k \coloneqq \omega \sqrt{\epsilon} \sqrt{\mu} \in \mathbb{C}$ with $\Re[k] \, \Im[k]=0$, where $\omega$ is the radiation frequency, and $\Re$ and $\Im$ stand for the real and imaginary parts, respectively. The generally frequency-dependent constitutive parameters are taken at a given central frequency. Given that $ k_0 \coloneqq \omega\sqrt{\epsilon_{0}}\sqrt{\mu_{0}}$ is the propagation constant in vacuum, one identifies the following cases: (\textit{i}) $k>k_{0} $ for ordinary materials, (\textit{ii}) $0<k<k_{0}$ for double-positive (DPS) metamaterials, (\textit{iii}) $k<0$ for double-negative (DNG) metamaterials, (\textit{iv}) $k\in \mathbb{C}$ for single-negative (SNG) (meta)materials, and (\textit{v}) $k=0$ for nihility metamaterials. This widely used, convenient classification is certainly unrealistic for \textit{passive} media. Indeed, for such media to exhibit a causal response to electromagnetic excitation their constitutive parameters must be complex functions of the radiation frequency that satisfy the classical Kramers-Kronig relations \cite{Jackson_book}. These relations show that the real and imaginary parts of the constitutive parameters are Hilbert transforms of each other and so none of them can vanish identically. It is possible for \textit{active} metamaterials, however, to be causal without satisfying these relations \cite{NistadSkaar2008,Srivastava2021}. Thus, the above classification is not unrealistic for this type of media. Active metamaterials attract considerable interest due to their potential applications \cite{HajianGhobadiButunEtAl2019,XiaoWangLiuEtAl2020,KoenderinkMonticoneStaudeEtAl2021,YouLanMaEtAl2023}. 
  
\subsection{Radiated Fields}
A time-harmonic, electric-current volume density $\mathbf{J}(\mathbf{r},t)\in L^2\left(V;\mathbb{C}^3\right)$ (herein referred to simply as ``the source''), generates an electric field $\mathbf{E}(\mathbf{r},t)$. Outside $V$, the electric field $\mathbf{E}(\mathbf{r})$ may be expressed as a multipole expansion, namely
\begin{equation}
\mathbf{E}(\mathbf{r})=\sum_{l,m}  \mbox{\boldmath${\nabla}$}\times\lbrack h_{l}^{(+)}(k_{0}r)\mathbf{Y}_{l,m}(\hat{\mathbf{r}})] \, a_{l,m}^{(1)}+ ik_{0}h_{l}^{(+)}(k_{0}r)\mathbf{Y}_{l,m}(\hat{\mathbf{r}}) \, a_{l,m}^{(2)},\quad\mathbf{r}\notin V,
\label{Eq:multipole_expansion_1}%
\end{equation}
wherein $a_{l,m}^{(j)}\in \mathbb{C}$ are the multipole moments of the radiated field, $\hat{\mathbf{r}}\coloneqq\mathbf{r}/r$, $h_{l}^{(+)}$ denotes the spherical Hankel function of the first kind and order $l$, corresponding to outgoing spherical waves in the far zone, $\mathbf{Y}_{l,m}$ is the vector spherical harmonic of degree $l$ and order $m$, and $j=1$ and $j=2$ correspond to electric and magnetic multipole fields, respectively. The index $l\in\mathbb{N}^*$ is sometimes referred to as the multipole order of the field; $m=-l,-l+1,...,0,...,l-1,l$. Practically, though, the multipole expansion in (\ref{Eq:multipole_expansion_1}) is truncated at the multipole order $l_{\text{max}}$. That is because fields that differ by less than a prescribed error are essentially indistinguishable. The value of $l_{\text{max}}$ is determined by the spatial extent of the source and the maximum spatial frequency components, and is related to the degrees of freedom of the radiated field, which is equal to the Nyquist number \cite{BucciFranceschetti1987,BucciFranceschetti1989,BucciGennarelliSavarese1998}. It has been shown that an empirical value of $l_{\text{max}}$ which yields very good results is given by $l_{\text{max}}\approx ka+1.8 \, d^{2/3} (ka)^{1/3}$, where $d$ is the number of digits of accuracy \cite{SongChew2001}. Sometimes it suffices to take $l_{\text{max}} \approx ka $. In fact, the spherical Bessel functions reach their maximum amplitudes when their argument is approximately equal to their order and for large orders the maximum is sharply peaked then. The multipole moments $a_{l,m}^{(j)}$ are \emph{uniquely} determined by the projections of $\mathbf{E}(\mathbf{r})$ onto a sphere of radius $R>a$. One finds that the multipole moments are given by
\begin{equation}
a_{l,m}^{(j)}=\left\{
\begin{array}
[c]{l}%
[i l (l+1) k_{0}h_{l}^{(+)}(k_{0}R)]^{-1} \int_{\Omega} \overline{\mathbf{Y}}_{l,m}(\hat{\mathbf{r}}) \cdot \mathbf{E}(R\hat{\mathbf{r}})d\hat{\mathbf{r}}\quad;j=1

\\

[l(l+1)k_{0}v_{l}(k_{0}R)]^{-1} \int_{\Omega} \hat{\mathbf{r}} \times \overline{\mathbf{Y}}_{l,m} \cdot \mathbf{E}(R\hat{\mathbf{r}})d\hat{\mathbf{r}}\quad;j=2,
\end{array}
\right.
\label{MultipoleMoments2}
\end{equation}
where the overline denotes the complex conjugate.
\subsection{Radiated Field Prescription Condition}
The source radiates a prescribed field outside $V$. The prescription condition is expressed as $a_{l,m}^{(j)}=(\mathcal{B}_{l,m}^{(j)},\mathbf{J})$. For piecewise-constant, radially-symmetric propagation media with constitutive parameters given by $f(\mathbf{r})$ we obtain 
\begin{equation}
\mathcal{B}_{l,m}^{(j)}\coloneqq\left\{
\begin{array}
[c]{l}
\frac{-\eta_{0}}{l(l+1)} \overline{F}_{l}^{(1)}(x_{0},x) \, \mbox{\boldmath${\nabla}$} \times \left[  j_{l}(\overline{k}r)\mathbf{Y}_{l,m}\right]\quad;j=1
\\
\frac{-ik_{0}\eta_{0}}{l(l+1)} \overline{F}_{l}^{(2)}(x_{0},x) j_{l}(\overline{k}r) \, \mathbf{Y}_{l,m}(\hat{\mathbf{r}
})\quad;j=2,
\end{array}
\right.
\label{eq_lor17d}
\end{equation}
where $j_{l}$ is the spherical Bessel function of the first kind and order $l$, and $x_{0}$ and $x$ are defined as $x_{0} \coloneqq k_{0}a\coloneqq\omega\sqrt{\epsilon_{0}\mu_{0}}a=2\pi\frac{a}{\lambda_{0}}$, and $x \coloneqq k a\coloneqq\omega\sqrt{\epsilon}\sqrt{\mu}a\coloneqq x\sqrt{\epsilon_{r}}\sqrt{\mu_{r}}=2\pi\frac{a}{\lambda}$. The relative electric permittivity and magnetic permeability are $\epsilon_{r}\coloneqq\epsilon\epsilon_{0}^{-1}$ and $\mu_{r}\coloneqq\mu\mu_{0}^{-1}$, respectively, and $\lambda$ is the wavelength. Explicit expressions for the complex Mie amplitudes $F_{l}^{(1)}(x_{0},x)$ and $F_{l}^{(2)}(x_{0},x)$ are given in \ref{Appendix:A}.

\subsection{The Constrained Minimization Problem}
Let $\mathcal{E}$ be the ``source energy'' defined as the $L^{2}$-norm of the current density, i.e.,
\begin{equation}
\mathcal{E} \coloneqq \left\langle \mathbf{J},\mathbf{J}\right\rangle \coloneqq \int \left|\mathbf{J}(\mathbf{r})\right|^{2}\,d\mathbf{r},
\end{equation}
and let $\Im[\mathcal{P}]$ be the reactive power defined as
\begin{equation}
\Im[\mathcal{P}] \coloneqq \Im[-\frac{1}{2}(\mathbf{J}(\mathbf{r}),\int d\mathbf{r^{\prime}\underline{\mathbf{G}}}(\mathbf{r},\mathbf{r^{\prime}})\cdot\mathbf{J}(\mathbf{r^{\prime}}))], 
\end{equation}
where $\Im$ stands for the imaginary part and $\underline{\mathbf{G}}(\mathbf{r},\mathbf{r^{\prime}})$ is the dyadic Green's function. For a radiating spherical source embedded in vacuum, $\underline{\mathbf{G}}(\mathbf{r},\mathbf{r^{\prime}})$ can be expressed as
\begin{align}
\underline{\mathbf{G}}(\mathbf{r,r}^{\prime})  &  =\sum_{l,m}\frac{-\eta_{0}}{\mu_{r} l\left(l+1\right)}\bigg\{F_l^{(2)}\left[ k^{2} j_{l}(k r_{<})\mathbf{Y}_{l,m}(\hat{\mathbf{r}}_{<})\right]  \left[  {h_{l}^{(+)}}(k_{0}r_{>})\overline{\mathbf{Y}}_{l,m}(\hat{\mathbf{r}}_{>})\right]  \bigg. \nonumber
\\
& +\bigg. F_l^{(1)}\mbox{\boldmath${\nabla}$}\times\left[j_{l}(k r_{<})\mathbf{Y}_{l,m}(\hat{\mathbf{r}}_{<})\right]\mbox{\boldmath${\nabla}$}\times\left[  {h_{l}^{(+)}}(k_{0}r_{>})\overline{\mathbf{Y}}_{l,m}(\hat{\mathbf{r}}_{>})\right]  \bigg\}.
\label{Green_Function}
\end{align}
The $<$ ($>$) subscript designates the smaller (larger) of $r$ and $r^{\prime}$. 

The problem under consideration is
\begin{gather}
\min\limits_{\mathbf{J}\in X} \left\langle \mathbf{J},\mathbf{J}\right\rangle,
\label{optim-problem-1}
\\
X\coloneqq\left\{\mathbf{J}\in L^{2}\left(V;\mathbb{C}^{3}\right):\left\langle \mathcal{B}_{l,m}^{(j)},\mathbf{J}\right\rangle=a_{l,m}^{(j)},\;\Im\left[\mathcal{P}\right]=0\right\}.
\label{optim-problem-2}
\end{gather}
This corresponds to the problem of reconstructing a radiating electromagnetic source embedded in a non-homogeneous background with generalized constitutive parameters by minimizing its $L^{2}$-norm subject to a prescribed radiated field and a vanishing reactive power. As indicated earlier, the minimization of the source's $L^{2}$-norm constitutes a useful criterion for the minimization of the actual physical resources of the radiating source. In addition, requiring the vanishing of the reactive power is desirable in radiating systems as it corresponds to the vanishing of the ``useless'' power.
\subsection{Existence of Solutions}
Theorems \ref{equivalence_thm} and \ref{auxiliary_thm} establish the existence of a solution to the constrained optimization problem (herein referred to as problem~(\ref{optim-problem-1},\ref{optim-problem-2})).

\begin{thm}
\label{equivalence_thm}

Problem (\ref{optim-problem-1},\ref{optim-problem-2}) is equivalent to the auxiliary problem 
\begin{gather}
\min\limits_{\mathbf{J}\in X\cap\overline{B_{\Gamma}\left(\mathbf{J}_{0}\right)}}\mathcal{E}\left(\mathbf{J}\right), \label{NEWoptim-problem}
\\
\overline{B_{\Gamma}\left(  \mathbf{J}_{0}\right)}=\left\{\mathbf{J}\in L^{2}\left(V;\mathbb{C}^{3}\right):\|\mathbf{J}-\mathbf{J}_{0}\|\leq\Gamma \right\}.
\label{NEWoptim-problem2}
\end{gather}
\end{thm}
 
\begin{proof}
Since $X$ is a closed (unbounded, and nonconvex) subset of a normed vector space and since $\mathcal{E}$ is a coercive functional, then
there exist $\mathbf{J}_{0}\in X$ and $\Gamma>0$
such that
\begin{equation}
\inf\limits_{\mathbf{J}\in X}\mathcal{E}\left(  \mathbf{J}\right)
=\inf\left\{  \mathcal{E}\left(  \mathbf{J}\right)  :\mathbf{J}\in
X\cap\overline{B_{\Gamma}\left(  \mathbf{J}_{0}\right)  }\right\} .
\label{auxiliary}%
\end{equation}
\end{proof}
\begin{thm}
\label{auxiliary_thm}
There exists a solution to the auxiliary problem~(\ref{NEWoptim-problem},\ref{NEWoptim-problem2}).
\end{thm}
\begin{proof}
Given that 
\begin{enumerate}
	\item $\mathcal{E}$ is a weakly sequentially lower semi-continuous functional, and
	\item $X\cap\overline{B_{\Gamma}\left(  \mathbf{J}_{0}\right)}$ is a weakly sequentially compact subset of a Hilbert space
\end{enumerate}
 there exists, by virtue of the generalized Weierstrass theorem, at least one solution to problem~(\ref{NEWoptim-problem},\ref{NEWoptim-problem2}). Consequently, there exists at least one solution to problem~(\ref{optim-problem-1},\ref{optim-problem-2}). 
\end{proof}
%
\subsection{Explicit Form of Solutions}
Let $\mathbf{J}_{\mathcal{E},\mathcal{P}}$ be a minimizer. Using the Lagrange multiplier method, it can be shown that 
\begin{equation}
\mathbf{J}_{\mathcal{E},\mathcal{P}}(\mathbf{r})=\sum_{j,l,m}\frac{a_{l,m}^{(j)}}{\left\langle \mathcal{B}_{l,m}^{(j)},\mathcal{D}_{l,m}^{(j)}\right\rangle}\mathcal{D}_{l,m}^{(j)}(\mathbf{r}),
\label{eqPP70}%
\end{equation}
where
\begin{equation}
\mathcal{D}_{l,m}^{(j)}(\mathbf{r})=\left\{
\begin{array}
[c]{l}
-\frac{\eta_{0}}{l(l+1)}\mbox{\boldmath${\nabla}$}\times\lbrack j_{l}(K_{\chi}r)\mathbf{Y}_{l,m}]\quad;j=1
\\
-\frac{i\eta_{0}K}{l(l+1)}j_{l}(K_{\chi}r)\mathbf{Y}_{l,m}(\hat{\mathbf{r}}
)\quad;j=2,
\end{array}
\right.
\label{eqPP52}
\end{equation}
and
\begin{equation}
\left\langle \mathcal{B}_{l,m}^{(j)},\mathcal{D}_{l,m}^{(j)}\right\rangle=\left\{
\begin{array}
[c]{l}%
\eta_{0}^{2}F_{l}^{\left(  1\right)  }\int_{0}^{a}dr\left[  j_{l}%
(kr)j_{l}(K_{\chi}r)+\frac{kK_{\chi}r^{2}}{l(l+1)} u_{l}\left(  kr\right)  u_{l}\left(
K_{\chi}r\right)  \right];j=1
\\
\eta_{0}^{2}F_{l}^{\left(  2\right)  }\frac{k_{0}K_{\chi}}{l(l+1)}\int_{0}^{a}%
drr^{2}j_{l}(kr)j_{l}(K_{\chi}r)\quad;j=2.
\end{array}
\right.
\label{BD_product}
\end{equation}
Radial functions $u_{l}$ are defined in \ref{Appendix:A}. The quantity $K_{\chi}=\sqrt{{k^{2}-\chi\mu}{\omega}}$, where $\chi\in\mathbb{R}$ is the Lagrange multiplier, is a modified propagation constant which appears in the wave equation $\mbox{\boldmath${\nabla}$}\times\mbox{\boldmath${\nabla}$}\times \mathbf{J}_{\mathcal{E},\mathcal{P}}(\mathbf{r})-K_{\chi}^{2}\mathbf{J}_{\mathcal{E},\mathcal{P}}(\mathbf{r})=\mathbf{0}$. This equation governs the spatiotempral variation of the source's optimized current distribution. It is obtained with the assumption that $k^2 \in \mathbb{R}$. This condition is satisfied by DPS, DNG, SNG, and nihility media. To have a time-harmonic current distribution transmitting a time-harmonic electromagnetic field through the source's substrate, we need to have $K_{\chi}^2>0$. 

It can also be shown that the minimum energy of the tuned source is given by 
\begin{equation}
\mathcal{E}_{\mathcal{E},\mathcal{P}}=\sum_{j,l,m} T_{l,m}^{(j)} \, |a_{l,m}^{(j)}|^{2}, 
\label{Tuned_Energy}
\end{equation}
where
\begin{equation}
T_{l,m}^{(j)} \coloneqq \frac{\left\langle \mathcal{D}_{l,m}^{(j)},\mathcal{D}_{l,m}^{(j)}\right\rangle}{\left\langle \mathcal{B}_{l,m}^{(j)},\mathcal{D}_{l,m}^{(j)}\right\rangle^2}
\label{Rjlm-Definition}
\end{equation}  
and
\begin{equation}
\left\langle \mathcal{D}_{l,m}^{(j)},\mathcal{D}_{l,m}^{(j)}\right\rangle=\left\{
\begin{array}
[c]{l}%
\eta_{0}^{2}  \int_{0}^{a}\left[ \left\vert j_{l}(K_{\chi} r)\right\vert ^{2}+\frac{\left\vert K_{\chi} \, r \, u_{l}(K_{\chi} r)\right\vert ^{2}}{l\left(  l+1\right)  }\right]dr  \quad;j=1
\\
\frac{\eta_{0}^{2}|K_{\chi}|^{2}}{l(l+1)} \int_{0}^{a}\left\vert r
j_{l}(K_{\chi} r)\right\vert ^{2}dr \quad;j=2.
\end{array}
\right.
\label{DD_product}
\end{equation}
In terms of the field multipole moments, the explicit expression of the complex interaction power reads 
\begin{equation}
\mathcal{P=}\sum_{j,l,m}   \frac{i}{2\chi} \left[ \frac{\left\langle \mathcal{D}_{l,m}^{(j)},\mathcal{D}_{l,m}^{(j)}\right\rangle}{\left\vert \left\langle \mathcal{B}_{l,m}^{(j)},\mathcal{D}_{l,m}^{(j)}\right\rangle \right\vert ^{2}}+\gamma_{l,m}^{(j)}\right]  |a_{l,m}^{(j)}|^{2}, 
\label{CompIntPower1}
\end{equation}
while the reactive power is given by
\begin{equation}
\Im\left[  \mathcal{P}\right]  =\sum_{j,l,m}   \frac{1}{2\chi} \left[ \frac{\left\langle \mathcal{D}_{l,m}^{(j)},\mathcal{D}_{l,m}^{(j)}\right\rangle}{ \left\vert \left\langle \mathcal{B}_{l,m}^{(j)},\mathcal{D}_{l,m}^{(j)}\right\rangle  \right\vert ^{2}}
+\Re\left[\gamma_{l,m}^{(j)}\right] \right]   |a_{l,m}^{(j)}|^{2}, 
\label{eqU15}
\end{equation}
where
\begin{equation}
\gamma_{l,m}^{(j)}=\left\{
\begin{array}
[c]{l}
\left[\,\overline{F_{l}^{(1)} k \, u_{l}(k a)}\,\right]^{-1} \left[-i\frac{k_{0}}{\eta_{0}}\chi l(l+1)v_{l}(k_{0}a)-\frac{K_{\chi} u_{l}(K_{\chi}a)}{ \left\langle \mathcal{B}_{l,m}^{(1)},\mathcal{D}_{l,m}^{(1)}\right\rangle} \right];j=1
\\ 
\left[\,\overline{F_{l}^{(2)}k_{0}j_{l}(k a)}\,\right]^{-1} \left[-i\frac{k_{0}}{\eta_{0}}\chi l(l+1)h_{l}^{(+)}(k_{0}a)-\frac{K_{\chi}j_{l}(K_{\chi}a)}{ \left\langle   \mathcal{B}_{l,m}^{(2)},\mathcal{D}_{l,m}^{(2)}\right\rangle} \right]; j=2.
\end{array}
\right.
\label{eqU1}
\end{equation}

\section{Results}
\label{Sec_Results}
In this section, we establish the new results of this paper, namely, the lower boundedness, global minimality, and uniqueness of the tuned minimal source. We also present a numerical study to give an intuitive understanding of the properties of this source in a few typical cases.
\subsection{Boundedness}
\begin{thm}
The untuned minimum source energy $\mathcal{E}_{\mathcal{E}}$ is a lower bound on the tuned minimum source energy $\mathcal{E}_{\mathcal{E},\mathcal{P}}$ for source substrates with generalized constitutive parameters $\left(\epsilon,\mu\right) \in \mathbb{R}^2$.
\label{Theorem_Boundedness}
\end{thm}
\begin{proof}
The necessary and sufficient condition for $\mathcal{E}_{\mathcal{E}}$ to be a unique lower bound on $\mathcal{E}_{\mathcal{E},\mathcal{P}}$ may be expressed as
\begin{equation}
\mathcal{E}_{\mathcal{E}} \leq \mathcal{E}_{\mathcal{E},\mathcal{P}}.
\label{Lower_Bound_Condition-1}
\end{equation}
Substituting the expressions of $\mathcal{E}_{\mathcal{E}}$ and $\mathcal{E}_{\mathcal{E},\mathcal{P}} $ from (\ref{Tuned_Energy}) yields
\begin{equation}
\sum_{j,l,m}\left(\left.T_{l,m}^{(j)}\right|_{\chi}-\left.T_{l,m}^{(j)}\right|_{0}\right)\left|a_{l,m}^{(j)}\right|^{2} \geq 0. 
\label{Lower_Bound_Condition-2}
\end{equation}
A sufficient condition for this to hold is that
\begin{equation}
\left.T_{l,m}^{(j)}\right|_{\chi}-\left.T_{l,m}^{(j)}\right|_{0} \geq 0, \quad \forall j,l,m.
\label{Lower_Bound_Condition-3}
\end{equation}
Substituting the explicit expression for $T_{l,m}^{(j)}$ (from (\ref{Rjlm-Definition}), (\ref{BD_product}), and (\ref{DD_product})) yields the following conditions
\begin{subnumcases}
{}
$ $\left[\int_{0}^{a} \left( \left| j_{l}(k r) \right|^{2} + \frac{\left| k\, r\, u_{l}(k r) \right|^{2}}{l(l+1)} \right) dr \right]\left[\int_{0}^{a} \left( \left| j_{l}(K_{\chi} r) \right|^{2} + \frac{\left| K_{\chi}\, r\, u_{l}(K_{\chi} r) \right|^{2}}{l(l+1)} \right) dr \right] $ $ \nonumber 
\\
$ $ \qquad \qquad \qquad \qquad \quad \geq \left| \int_{0}^{a} \left( j_{l}(k r) j_{l}(K_{\chi} r) + \frac{k K_{\chi}\, r^{2}\, u_{l}(k r) u_{l}(K_{\chi} r)}{l(l+1)} \right) dr \right|^{2};$ $ & $j=1$ \label{Lower_Bound_Condition-4}
\\[12pt]
$ $\left[ \int_{0}^{a} \left| r\, j_{l}(k r) \right|^{2} dr \right]\left[ \int_{0}^{a} \left| r\, j_{l}(K_{\chi} r) \right|^{2} dr \right] \geq \left| \int_{0}^{a} j_{l}(k r)\, j_{l}(K_{\chi} r)\, r^{2} dr\right|^{2};$ $ & $j=2$ \label{Lower_Bound_Condition-5}
\end{subnumcases}

Inequality~(\ref{Lower_Bound_Condition-5}) is satisfied by virtue of the Cauchy-Schwarz theorem, since the integrals are finite Lommel integrals (See equations~(\ref{Lower_Bound_Condition-6a}) and (\ref{Lower_Bound_Condition-6b})). As for (\ref{Lower_Bound_Condition-4}), it is more instructive to use its original form, namely,
\begin{equation}
\begin{aligned}
&\left[\int_{\Omega}d\hat{\mathbf{r}}\int_{0}^{a} \left|\mbox{\boldmath${\nabla}$} \times \left[ j_{l}(k r)\mathbf{Y}_{l,m}\right] \right|^2dr \right]       \left[\int_{\Omega}d\hat{\mathbf{r}}\int_{0}^{a} \left|\mbox{\boldmath${\nabla}$} \times \left[ j_{l}(K_{\chi} r)\mathbf{Y}_{l,m}\right] \right|^2 dr \right]
\\
& \qquad \qquad \qquad \geq \left| \int_{\Omega}d\hat{\mathbf{r}}\int_{0}^{a} \overline{\mbox{\boldmath${\nabla}$} \times \left[ j_{l}(k r)\mathbf{Y}_{l,m}\right]} \cdot \mbox{\boldmath${\nabla}$} \times \left[ j_{l}(K_{\chi} r)\mathbf{Y}_{l,m}\right]dr \right|^2.
\label{Lower_Bound_Condition-7}
\end{aligned}
\end{equation}
This expression clearly shows that we have another inequality of the Cauchy-Schwarz type provided that  
\begin{equation}
\int_{\Omega}d\hat{\mathbf{r}} \int_{0}^{a} \left|\mbox{\boldmath${\nabla}$} \times \left[ j_{l}(\alpha r)\mathbf{Y}_{l,m}\right]\right|^{2}<\infty.
\label{Lower_Bound_Condition-8}
\end{equation}
For this square-integrability criterion to be satisfied, it is sufficient to show that the spherical Bessel functions $j_{l}$ and their combinations $(r \, u_{l})$ are square-integrable, that is,  
\begin{equation}
\int_{0}^{a} \left| j_{l}(\alpha r) \right|^{2} dr <\infty
\label{Lower_Bound_Condition-10-1}
\end{equation}
and 
\begin{equation}
\int_{0}^{a} \left| r \, u_{l}(\alpha r)\right|^{2} dr < \infty. 
\label{Lower_Bound_Condition-10-2}
\end{equation}
These conditions involve definite integrals of the form
\begin{equation}
\int_0^a r^s j_{l}(\alpha r) j_{l'}(\alpha' r) dr
\label{Eq:results_22}
\end{equation}
where $l,l^{'} \in \mathbb{N}^{*}$, $\alpha,\alpha' \in \mathbb{C}$, and $s=0,2$ (see \ref{Appendix:B}). Closed-form expressions are not necessary to establish (\ref{Lower_Bound_Condition-10-1}) and (\ref{Lower_Bound_Condition-10-2}). These conditions simply result from the fundamental theorem of calculus for complex functions, since the integrand in (\ref{Eq:results_22}) is the product of entire functions, i.e., complex-valued functions that are holomorphic, and thus continuous, on the whole complex plane \cite{Cartan1995}. Other useful properties of the integrand are discussed in \ref{Appendix:B}. 

The equality in (\ref{Lower_Bound_Condition-4}), (\ref{Lower_Bound_Condition-5}), and (\ref{Lower_Bound_Condition-7}) holds if and only if $j_{l}(k r)=j_{l}(K_{\chi} r)$, \emph{i.e.}, when $0 \in \Xi$, where $\Xi \coloneqq\left\{\chi\in\mathbb{R}:\Im\left[\mathcal{P}\right]=0\right\}$, i.e, it is the set of all the Lagrange multipliers that satisfy the tuning constraint.
\end{proof}
\subsection{Global Minimality and Uniqueness}
Let $\chi_0 \in \Xi$ be the element that satisfies
\begin{equation}
\left|\chi_{0}\right| =\inf\limits_{\chi\in\Xi}\left\{  \left|\chi\right| \right\}.
\label{Chi_Nought_Def}
\end{equation}
\begin{cor}
When the unique lower bound $\mathcal{E}_{\mathcal{E}}$ belongs to the set of tuned minimum source energies $\left.\mathcal{E}_{\mathcal{E},\mathcal{P}}\right|_{\chi}$, it is also the unique global minimum.
\label{Theorem_Global_Minimality_1}
\end{cor}
\begin{proof}
This corollary follows directly from the lower-boundedness of $\mathcal{E}_{\mathcal{E}}$ (Theorem \ref{Theorem_Boundedness}). 
\end{proof}
\begin{thm}
The tuned minimum source energy $\left.\mathcal{E}_{\mathcal{E},\mathcal{P}}\right|_{\chi_{0}}$ is the unique global minimum on the tuned minimum energy $\mathcal{E}_{\mathcal{E},\mathcal{P}}$ for source substrates with generalized constitutive parameters $\left(\epsilon,\mu\right) \in \mathbb{R}^2$ and for which the tuned source is a ``perturbation'' of the untuned source.
\label{Theorem_Global_Minimality_2}
\end{thm} 
\begin{proof}
The necessary and sufficient condition for $\left.\mathcal{E}_{\mathcal{E},\mathcal{P}}\right|_{\chi_{0}}$ to be a unique global minimum on $\mathcal{E}_{\mathcal{E},\mathcal{P}}$ may be expressed as
\begin{equation}
\left.\mathcal{E}_{\mathcal{E},\mathcal{P}}\right|_{\chi_{0}} < \left.\mathcal{E}_{\mathcal{E},\mathcal{P}}\right|_{\chi}.
\label{Global_Min_Condition-1}
\end{equation}
To establish (\ref{Global_Min_Condition-1}) we expand $\left.\mathcal{E}_{\mathcal{E},\mathcal{P}}\right|_{\chi}$ in formal Taylor series about $\chi=0$.
The series expansion yields
\begin{equation}
\left.\mathcal{E}_{\mathcal{E},\mathcal{P}}\right|_{\chi} = c_{0}+c_{1} \chi+c_{2} \chi^2+O[\chi^3], 
\label{Global_Min_Condition-7_1}
\end{equation}
where
\begin{equation}
c_{k} \coloneqq \frac{1}{k!} \left.\frac{\partial ^k \mathcal{E}_{\mathcal{E},\mathcal{P}}}{\partial \chi^k}\right|_{0}=\frac{1}{k!} \sum_{j,l,m}\left.\frac{\partial ^k T_{l,m}^{(j)}}{\partial \chi^k}\right|_{0}\left|a_{l,m}^{(j)}\right|^{2}, \quad k=0,1,2.
\end{equation}
where, in the last step, we have used (\ref{Tuned_Energy}). The higher-order terms $O\left(\chi ^3\right)$, i.e., the remainder of the expansion, can be expressed in the Lagrange form as
\begin{equation}
    O\left(\chi ^3\right) = \frac{1}{3!} \left.\frac{\partial ^3 \mathcal{E}_{\mathcal{E},\mathcal{P}}}{\partial \chi^3} \right|_{\xi} \chi^3,
\end{equation}
where $\xi$ is between 0 and $ \chi$. Due to the intricate form of the third derivative, no attempt will be made here to discuss the boundedness of $O\left(\chi ^3\right)$. In what follows, we shall assume that it is negligible compared to the main term of the series. Systems for which this would be a valid assumption are systems for which the tuned source is a ``perturbation'' (in this sense) of the untuned source. This hypothesis is supported by the numerical study in \cite{MKB_SIAM_2008}. Now, $\left.\mathcal{E}_{\mathcal{E},\mathcal{P}}\right|_{0}$ is a lower bound on $\left.\mathcal{E}_{\mathcal{E},\mathcal{P}}\right|_{\chi} $ (theorem \ref{Theorem_Boundedness}). It follows that $\chi=0$ is a local minimum of $\left.\mathcal{E}_{\mathcal{E},\mathcal{P}}\right|_{\chi}$  with respect to $\chi$. Therefore, 
\begin{equation}
\left.\frac{\partial \mathcal{E}_{\mathcal{E},\mathcal{P}}}{\partial \chi}\right|_{0}=0 \Leftrightarrow c_{1} = 0
\end{equation}
and
\begin{equation}
\left.\frac{\partial ^2 \mathcal{E}_{\mathcal{E},\mathcal{P}}}{\partial \chi^2}\right|_{0} \geq 0 \Leftrightarrow c_{2} \geq 0,
\end{equation}
which establishes (\ref{Global_Min_Condition-1}).
\end{proof}
\subsection{Discussion}
In this section we briefly discuss the properties of a few typical solutions of the problem. In Fig.~\ref{fig:React_Power}, contour plots of $\Im\left[  \mathcal{P}\right]  =0$ are shown for sources with electrical sizes $k_{0}a=\pi/4, \pi/2, \pi$, and $2\pi$. The plots illustrate the conditions on permittivity and permeability under which the reactive power vanishes. Several general features are immediately apparent. First, as the electrical size of the source decreases—that is, as the ratio of the physical source size to the wavelength of emitted radiation becomes smaller—the range of permittivity and permeability values that yield zero reactive power also narrows. This is in line with the established fact that tuning and impedance matching electrically large antennas is easier due to their reduced reactive components. Second, for an active metamaterial substrate, the source tuning could remain stable during operation if the substrate’s permittivity and permeability values vary along specific curves within the reactive power plot. Also, for any given permittivity (or permeability), there exists a discrete, potentially large set of permeability (or permittivity) values that allow for effective source tuning. Lastly, the observed tuning behavior seems to favor DPS and DNG substrates over SNG substrates. If this feature persists in a more realistic treatment where, among other things, the $L^{2}$-norm of the current density—commonly referred to as ``source energy''—is replaced by the minimization of the actual energy, then one might speculate that an interplay of physical phenomena strongest in DPS and DNG materials is responsible for this effect. These include photonic band gaps, backward wave propagation, and electromagnetically induced transparency. In contrast, plasmon-polariton interactions—most pronounced in epsilon-negative (ENG) materials—and magnetic surface polaritons—most pronounced in mu-negative (MNG) materials—are likely to play a minor role. The numerical study shows, though, that the zero-reactive power curves can extend into the second quadrant of the plots (where $\epsilon < 0$) rather than the fourth (where $\mu < 0$).

\begin{figure}
    \centering
    \subfigure[]{\includegraphics[width=0.49\textwidth]{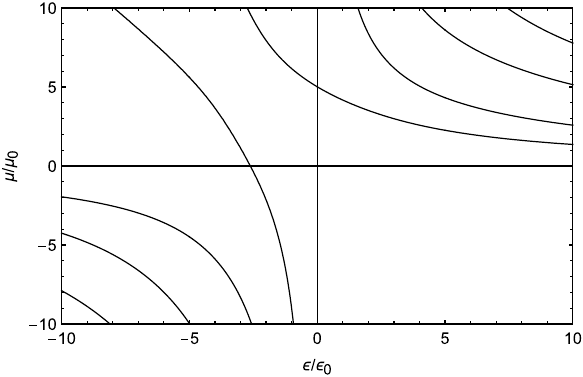}} 
    \subfigure[]{\includegraphics[width=0.49\textwidth]{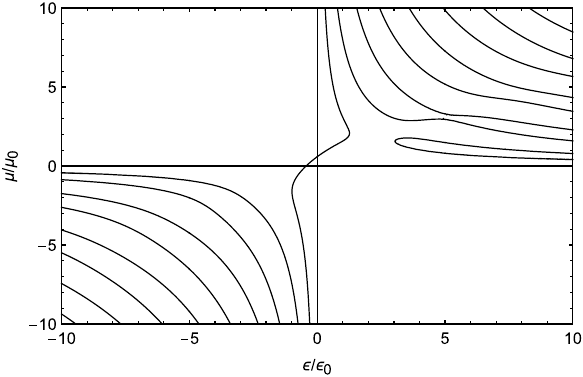}}
    \subfigure[]{\includegraphics[width=0.49\textwidth]{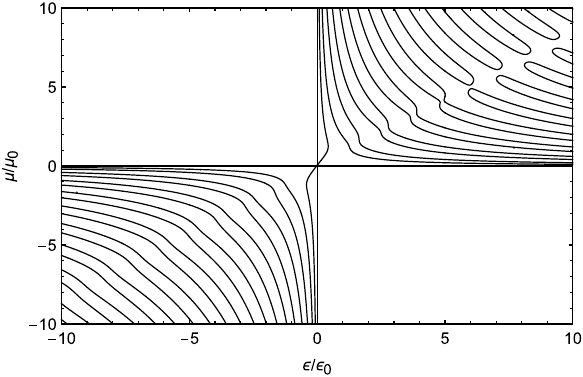}}
    \subfigure[]{\includegraphics[width=0.49\textwidth]{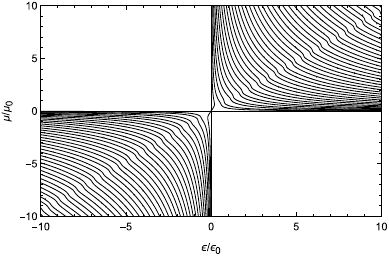}}
    \caption{Contour plots of $\Im\left[  \mathcal{P}\right]  =0$ for sources with (a) $k_{0}a=\pi/4$, (b) $k_{0}a=\pi/2$, (c) $k_{0}a=\pi$, and (d) $k_{0}a=2\pi$.}
    \label{fig:React_Power}
\end{figure}

\section{Conclusion}
This work has established global minimality and uniqueness for electromagnetic inverse source problems in complex, non-homogeneous media with generalized constitutive parameters. The initial question concerning the existence of a solution was answered in \cite{MKB_SIAM_2008}. The current density was reconstructed by minimizing its $L^{2}$-norm subject to a prescribed radiated field and a vanishing reactive power. The problem is both intriguing and challenging due to its general formulation, which requires minimizing an objective functional with nonconvex functional constraints on an unbounded domain. Numerical study reveals that for active metamaterial substrates, tuning could be maintained as permittivity and permeability vary along specific zero-reactive power curves. Also, for any given permittivity (or permeability), a range of corresponding dual parameters allows effective tuning. The observed tuning behavior appears to favor DPS and DNG substrates over SNG substrates, potentially due to dominant phenomena such as photonic band gaps, backward wave propagation, and electromagnetically induced transparency. Conversely, plasmon-polariton interactions—most pronounced in epsilon-negative (ENG) materials—and magnetic surface polaritons—most pronounced in mu-negative (MNG) materials—appear to have limited impact. This work fills an analytical gap by confirming boundedness, minimality, and uniqueness and provides insights into the practical tuning behavior of sources within metamaterials.

\section*{Disclosure of Interest}
The author reports that there are no competing interests to declare.

\section*{Data Availability Statement}
Data sharing is not applicable to this article as no new data were created or analyzed in this study.

\section*{Acknowledgment}
The author gratefully acknowledges partial support from Prince Mohammad Bin Fahd University and King Fahd University of Petroleum and Minerals.

\appendix

\section{Complex Mie amplitudes $F_{l}^{(1)}(x_{0},x)$ and $F_{l}^{(2)}(x_{0},x)$}
\label{Appendix:A}
The complex Mie amplitudes $F_{l}^{(1)}(x_{0},x)$ and $F_{l}^{(2)}(x_{0},x)$ introduced in (\ref{eq_lor17d}) are given by
\begin{align}
F_{l}^{(1)}(x_{0},x) &\coloneqq \frac{\sqrt{\mu_r}}{x_{0}x}\bigg\{\left[\sqrt{\epsilon_r}j_{l}\left(x\right)w_{l}\left(x_{0}\right)-\sqrt{\mu_r}u_{l}\left(x\right) y_{l}\left(x_{0}\right)\right]\bigg. \nonumber
\\
&+\bigg.i\left[\sqrt{\epsilon_r}j_{l}\left(x\right)u_{l}\left(x_{0}\right)-\sqrt{\mu_r}u_{l}\left(x\right)j_{l}\left(x_{0}\right)\right]\bigg\} \nonumber
\\
&\times \bigg\{\left[\sqrt{\epsilon_r}j_{l}\left(x\right)w_{l}\left(x_{0}\right)-\sqrt{\mu_r}u_{l}\left(x\right) y_{l}\left(x_{0}\right)\right]^{2} \bigg. \nonumber
\\
&+\bigg.\left[\sqrt{\epsilon_r}j_{l}\left(x\right)u_{l}\left(x_{0}\right)-\sqrt{\mu_r}u_{l}\left(x\right)j_{l}\left(x_{0}\right)\right]^{2}\bigg\}^{-1},
\label{Mie_Amplitude_j1}
\end{align}
and
\begin{align}
F_{l}^{(2)}(x_{0},x) &\coloneqq \frac{\mu_{r}\sqrt{\epsilon_r}}{x_{0}x} \bigg\{\left[\sqrt{\mu_r}j_{l}\left(x\right)w_{l}\left(x_{0}\right)-\sqrt{\epsilon_r}u_{l}\left(x\right) y_{l}\left(x_{0}\right)\right] \bigg. \nonumber
\\
&+\bigg.i\left[\sqrt{\mu_r}j_{l}\left(x\right)u_{l}\left(x_{0}\right)-\sqrt{\epsilon_r}u_{l}\left(x\right)j_{l}\left(x_{0}\right)\right] \bigg\} \nonumber
\\
&\times \bigg\{\left[\sqrt{\mu_r}j_{l}\left(x\right)w_{l}\left(x_{0}\right)-\sqrt{\epsilon_r}u_{l}\left(x\right) y_{l}\left(x_{0}\right)\right]^{2} \bigg.\nonumber
\\
&+\bigg.\left[\sqrt{\mu_r}j_{l}\left(x\right)u_{l}\left(x_{0}\right)-\sqrt{\epsilon_r}u_{l}\left(x\right)j_{l}\left(x_{0}\right)\right]^{2} \bigg\}^{-1},
\label{Mie_Amplitude_j2}
\end{align}
where the radial functions $u_{l},\,w_{l}$, and $v_{l}$ are defined as 
\begin{align}
&u_{l}(x) \coloneqq \frac{dj_{l}(x)}{d x}+\frac{j_{l}(x)}{x}=\frac{1}{2l+1}[(l+1)j_{l-1}(x)-l j_{l+1}(x)], 
\label{Eq:u_l}
\\
&w_{l}(x) \coloneqq \frac{dy_{l}(x)}{d x}+\frac{y_{l}(x)}{x}=\frac{1}{2l+1}[(l+1) y_{l-1}(x)-l y_{l+1}(x)],
\label{Eq:w_l}
\\
&v_{l}(x) \coloneqq \frac{dh_{l}^{(+)}(x)}{d x}+\frac{h_{l}^{(+)}(x)}{x}=u_{l}(x)+i w_{l}(x)
\label{Eq:v_l}
\end{align}
and where $y_l$, $h_{l}^{(+)}$ are the spherical Neumann and Hankel functions of order $l$, respectively. 

The spherical Bessel and Neumann functions are (\textit{i}) continuous on the closed upper half-plane $\left\{x\in \mathbb {C} \mid \Im(x)\geq 0\right\}$, (\textit{ii}) holomorphic in the upper half-plane $\left\{j_{l}(x),y_{l}(x)\in \mathbb {C} \mid \Im(x)>0\right\}$, and (\textit{iii}) real-valued when their arguments are real. Thus, by the Schwarz reflection principle \cite{Cartan1995}, they satisfy
\begin{equation}
\overline{j_{l}(x)}=j_{l}(\overline{x}) \quad \text{and} \quad \overline{y_{l}(x)}=y_{l}(\overline{x}).
\label{Eq:reflection1}
\end{equation}
Consequently, we also have
\begin{equation}
\overline{u_{l}(x)}=u_{l}(\overline{x}), \quad \overline{w_{l}(x)}=w_{l}(\overline{x}), \quad \text{and} \quad \overline{v_{l}(x)}=v_{l}(\overline{x}).
\label{Eq:reflection2}
\end{equation}
%
\section{Comments of the square-integrability conditions~(\ref{Lower_Bound_Condition-10-1}) and (\ref{Lower_Bound_Condition-10-2})}
\label{Appendix:B}

The square-integrability conditions  
\begin{equation}
\int_{0}^{a} \left| j_{l}(\alpha r) \right|^{2} dr <\infty
\label{Eq:A1}
\end{equation}
and 
\begin{equation}
\int_{0}^{a} \left| r \, u_{l}(\alpha r)\right|^{2} dr < \infty, 
\label{Eq:A2}
\end{equation}
(i.e., conditions~(\ref{Lower_Bound_Condition-10-1}) and (\ref{Lower_Bound_Condition-10-2}), respectively) are needed for the Cauchy-Schwarz inequality~(\ref{Lower_Bound_Condition-7}) to hold. As noted in section~\ref{Sec_Results}, these conditions involve definite integrals of the form
\begin{equation}
\int_0^a r^s \, j_{l}(\alpha r) \, j_{l'}(\alpha' r) \, dr 
\label{Eq:A3}
\end{equation}
where $l,l^{'} \in \mathbb{N}^{*}$, $\alpha,\alpha' \in \mathbb{C}$, $a \in \mathbb{R}$, and $s=0,2$. A useful property of these integrals obtained from the conjugate symmetry of the inner product and the Schwarz reflection principle (\ref{Appendix:A}) is that
\begin{equation}
\overline{\int_0^a r^s \, j_{l}(\alpha r) \, j_{l'}(\alpha' r) \, dr } = \int_0^a r^s \, \overline{j_{l}(\alpha r)} \, \overline{j_{l'}(\alpha' r)} \, dr = \int_0^a r^s \, j_{l}(\overline{\alpha} r) \, j_{l'}(\overline{\alpha'} r) \, dr. 
\end{equation}
These transformations allow us to rewrite the integrals in forms that come in helpful in the calculations.

The integral in (\ref{Eq:A1}) is recovered by setting $l=l'$, $\alpha'=\overline{\alpha}$, and $s=0$ in (\ref{Eq:A3}). 
%
Using (\ref{Eq:u_l}), we obtain the following explicit expression for the integral in (\ref{Eq:A2})
\begin{equation}
\begin{aligned}
\int_{0}^{a} \left| r \, u_{l}(\alpha r)\right|^2 dr &= \frac{l^2}{(2l+1)^2} \int_{0}^{a} r^2 j_{l+1}(\alpha r) j_{l+1}(\overline{\alpha} r) dr
\\
&+\frac{(l+1)^2 }{(2l+1)^2} \int_{0}^{a} r^2 j_{l-1}(\alpha r) j_{l-1}(\overline{\alpha} r) dr
\\
&-\frac{2 l (l+1)}{(2l+1)^2} \Re \left[ \int_{0}^{a} r^2 j_{l-1}(\alpha r) j_{l+1}(\overline{\alpha} r) dr\right].
\end{aligned}
\label{Eq:A5}
\end{equation}
The first two terms are Lommel integrals of the forms \cite{arfken,NIST_Handbook,WolframFunctions,Bowman_Book} 
\begin{subnumcases}
{\int_{0}^{a}\left|r j_{l}(\alpha r)\right|^2 dr =}
$ $ \frac{a^3}{2}\left[j_{l}(\alpha a)^2-j_{l+1}(\alpha a)j_{l-1}(\alpha a)\right] < \infty;$ $ & $\alpha=\overline{\alpha}$ \label{Lower_Bound_Condition-6a}
\\[12pt]
$ $\frac{a^2}{\alpha^2-\overline{\alpha}^2}\left[\overline{\alpha}j_{l}(\alpha a)j_{l-1}(\overline{\alpha} a)-\alpha j_{l}(\overline{\alpha} a)j_{l-1}(\alpha a)\right] < \infty;$ $ & $\alpha\neq\overline{\alpha}$ \label{Lower_Bound_Condition-6b}
\end{subnumcases}
so we need not worry about them. The third term is recovered from (\ref{Eq:A3}) with the substitutions $l \rightarrow l+1$, $l' \rightarrow l-1$, $\alpha'\rightarrow \overline{\alpha}$, and $s\rightarrow 2$. Now, the integrand in (\ref{Eq:A3}) is the product of entire functions, i.e., complex-valued functions that are holomorphic, and thus continuous, on the whole complex plane. Hence, by the fundamental theorem of calculus for complex functions \cite{Cartan1995}, the integral in (\ref{Eq:A3}) is finite, which establishes conditions~(\ref{Eq:A1}) and (\ref{Eq:A2}). 

It is instructive to see when the integral is real-valued. The spherical Bessel function of the first kind with complex argument $z$ and nonnegative, integer order $l$, $j_l(z)$, can be expressed as a power series of the form \cite{NIST_Handbook}
\begin{equation}
j_l(z) = z^l \sum_{n=0}^{\infty} \frac{\left(-1\right)^n z^{2n}}{2^n n! \left(2l+2n+1\right)!!} .
\label{Eq:A6}
\end{equation}
Thus, 
\begin{equation}
r^s j_l(\alpha r) j_{l'}(\alpha' r) = r^s \left(\alpha r \right)^l \sum_{n=0}^{\infty} \frac{\left(-1\right)^{n} \left(\alpha r\right)^{2n}}{2^{n} n ! \left(2l+2n +1\right)!!} \left(\alpha' r \right)^{l'} \sum_{n'=0}^{\infty} \frac{\left(-1\right)^{n'} \left(\alpha' r\right)^{2n'}}{2^{n'} n' ! \left(2l'+2n' +1\right)!!}
\label{Eq:A7-1}
\end{equation}
which is the Cauchy product of two convergent power series (since they represent entire functions). By separating the $n=n'$ and $n\neq n'$ terms, the terms in (\ref{Eq:A7-1}) can be rearranged as follows
\begin{equation}
\begin{aligned}
r^s j_l(\alpha r) j_{l'}(\alpha' r) &= \alpha^l \, \alpha'^{\,l'} r^{l+l'+s} \Bigg\{\sum_{n=0}^{\infty} \left[\frac{\left(\alpha \alpha'\right)^{n}  r^{2n}}{2^{n} n ! \left(2l+2n +1\right)!!}\right]^2\Bigg.
\\
&+ \Bigg.\sum_{n=0}^{\infty} \sum_{\substack{n'=0 \\ n'\neq n}}^{\infty}\frac{\left(-1\right)^{n+n'} \Re\left[\alpha^{2n} \, \alpha'\,^{2n'}\right] r^{2(n+n')}}{2^{n+n'-1} n! n'! \left(2l+2n +1\right)!!\left(2l'+2n' +1\right)!!}\Bigg\}.
\label{Eq:A10}
\end{aligned}
\end{equation}
Note that $l+l'+s \geq 2$, since, by definition $l,l^{'} \in \mathbb{N}^{*}$ and $s=0,2$. Thus
\begin{equation}
r^s j_l(\alpha r) j_{l'}(\alpha' r) \in \mathbb{R} \Rightarrow \sum_{n=0}^{\infty} \left[\frac{  r^{2n}}{2^{n} n ! \left(2l+2n +1\right)!!}\right]^2 \Im \left[ \alpha ^{2n+l}\alpha'^{\,2n+l'}\right]=0.
\label{Eq:A11}
\end{equation}
For the equation in (\ref{Eq:A11}) to hold independently of $r$, we need to have
\begin{equation}
\Im \left[ \alpha ^{2n+l}\alpha'^{\,2n+l'}\right]=0,
\label{Eq:A12}
\end{equation}
that is,
\begin{equation}
2n\left(\arg \alpha+\arg \alpha'\right)+l\arg \alpha+l'\arg \alpha'=0
\label{Eq:A13}
\end{equation}
where $\arg \alpha,\arg \alpha'$ are the principal values of the arguments of $\alpha$ and $\alpha'$, respectively. For (\ref{Eq:A13}) to be independent of $n$, the term in parentheses must vanish. We conclude that
the integral is real-valued if and only if $\arg \alpha=-\arg \alpha'$ and $l=l'$ (two obvious cases are when $\alpha,\alpha' \in \mathbb{R}$ or $\alpha'=\overline{\alpha}$).

\newpage

\bibliographystyle{iopart-num}

\bibliography{MTM}

\end{document}